\documentclass[runningheads]{llncs}
\usepackage[T1]{fontenc}
\usepackage{graphicx}
\usepackage{graphics}
\usepackage[ruled,linesnumbered]{algorithm2e}
\usepackage{amsmath}
\usepackage{tabularx}
\usepackage{environ}
\usepackage{amssymb}
\usepackage{todonotes}
\usepackage{tikz}
\usetikzlibrary{backgrounds}
\usepackage{caption}
\usepackage{subcaption}
\usetikzlibrary{positioning}
\usepgflibrary{arrows}
\usetikzlibrary{arrows,decorations.pathreplacing,calligraphy}
\usepackage{float}

\usepackage{thmtools} 
\usepackage{thm-restate}

\makeatletter
\newcommand{\problemtitle}[1]{\gdef\@problemtitle{#1}}
\newcommand{\probleminput}[1]{\gdef\@probleminput{#1}}
\newcommand{\problemquestion}[1]{\gdef\@problemquestion{#1}}
\NewEnviron{problemdescription}{
  \problemtitle{}\probleminput{}\problemquestion{}
  \BODY
  \par\addvspace{.5\baselineskip}
  \noindent
  \begin{tabularx}{\textwidth}{@{\hspace{\parindent}} l X c}
    \hline
    \multicolumn{2}{@{\hspace{\parindent}}l}{\@problemtitle} \\
    \textbf{Input:} & \@probleminput \\
    \textbf{Question:} & \@problemquestion\\
    \hline
  \end{tabularx}
  \par\addvspace{.5\baselineskip}
}
\makeatother

\newcommand{\sig}[1]{\Sigma^{#1}}
\newcommand{\plus}{\mathtt{+}}
\newcommand{\interval}[3]{#1[#2:#3]}

\newcommand{\len}[1]{|#1|}

\newcommand{\set}[1]{\lbrace #1 \rbrace}
\newcommand{\hdist}[2]{d_{\mathtt{HAM}}(#1,#2)}
\newcommand{\edist}[2]{d_{\mathtt{ED}}(#1,#2)}

\newcommand{\match}{\mathtt{Match}}
\newcommand{\misMatch}{\mathtt{MisMatch}}
\newcommand{\minMisMatch}{\mathtt{MinMisMatch}}
\newcommand{\alp}{\mathtt{alph}}
\newcommand{\var}{\mathtt{var}}
\newcommand{\term}{\mathtt{term}}

\newcommand{\regPat}{\mathtt{Reg}}
\newcommand{\oneRepPat}{\mathtt{1RepVar}}
\newcommand{\kRepPat}{\mathtt{kRepVar}}
\newcommand{\oneVarPat}{\mathtt{1Var}}
\newcommand{\nonCrossPat}{\mathtt{NonCross}}
\newcommand{\OV}{\mathtt{OV}}

\newcommand{\MS}{\mathtt{MS}}

\newcommand{\LCP}{\mathtt{LCP}}

\newcommand{\kScdPat}{\mathtt{kSCD}}

\newcommand{\kLocPat}{\mathtt{kLOC}}

\SetKwInput{KwData}{Input}
\SetKwInput{KwResult}{Output}
\begin{document}
\title{Matching Patterns with Variables Under Edit Distance}
%
%
\author{Pawe\l{}  Gawrychowski\inst{1}\orcidID{0000-0002-6993-5440} \and
Florin Manea\inst{2}\orcidID{0000-0001-6094-3324} \and
Stefan Siemer\inst{2}\orcidID{0000-0001-7509-8135}}
\authorrunning{Pawe\l{}  Gawrychowski \and Florin Manea \and Stefan Siemer}
%
\institute{University of Wroc\l{}aw, Faculty of Mathematics and Computer Science, Poland  \email{gawry@cs.uni.wroc.pl} \and
G\"ottingen University, Computer Science Department and CIDAS, Germany \email{\{florin.manea,stefan.siemer\}@cs.uni-goettingen.de}}
\maketitle              
\begin{abstract}

A pattern $\alpha$ is a string of variables and terminal letters. We say that $\alpha$ matches a word $w$, consisting only of terminal letters, if $w$ can be obtained by replacing the variables of $\alpha$ by terminal words. The matching problem, i.e., deciding whether a given pattern matches a given word, was heavily investigated: it is NP-complete in general, but can be solved efficiently for classes of patterns with restricted structure. If we are interested in what is the minimum Hamming distance between $w$ and any word $u$ obtained by replacing the variables of $\alpha$ by terminal words (so matching under Hamming distance), one can devise efficient algorithms and matching conditional lower bounds for the class of regular patterns (in which no variable occurs twice), as well as for classes of patterns where we allow unbounded repetitions of variables, but restrict the structure of the pattern, i.e., the way the occurrences of different variables can be interleaved. Moreover, under Hamming distance, if a variable occurs more than once and its occurrences can be interleaved arbitrarily with those of other variables, even if each of these occurs just once, the matching problem is intractable. In this paper, we consider the problem of matching patterns with variables under edit distance. We still obtain efficient algorithms and matching conditional lower bounds for the class of regular patterns, but show that the problem becomes, in this case, intractable already for unary patterns, consisting of repeated occurrences of a single variable interleaved with terminals. \looseness=-1

\end{abstract}

\section{Introduction}
A \emph{pattern with variables} is a string consisting of \emph{constant} or \emph{terminal letters} from a finite alphabet $\Sigma$ (e.g., $\mathtt{a,b,c}$), and \emph{variables} (e.g., $x, y, x_1, x_2$) from a potentially infinite set $\mathcal{X}$, with $\Sigma\cap \mathcal{X} = \emptyset$. In other words, a pattern $\alpha$ is an element of $PAT_{\Sigma} = (\mathcal{X}\cup\Sigma)^{\plus}$. A pattern $\alpha$ is mapped (by a function $h$ called substitution) to a word by substituting the variables by arbitrary strings of terminal letters; as such, $h$ simply maps the variables occurring in $\alpha$ to words over $\Sigma$. For example, $x x \mathtt{bbb} y y $ can be mapped to $\mathtt{aaaabbbbb}$ by the substitution $h$ defined by $(x \to \mathtt{aa}, y \to \mathtt{b})$. In this framework, $h(\alpha)$ denotes the word obtained by substituting every occurrence of a variable $x$ in $\alpha$ by $h(x)$ and leaving all the terminals unchanged. If a pattern $\alpha$ can be mapped to a string of terminals $w$, we say that $\alpha$ matches $w$; the problem of deciding whether there exists a substitution which maps a given pattern $\alpha$ to a given word $w$ is called the {\em (exact) matching problem.} \looseness=-1

\begin{problemdescription}
  \problemtitle{Exact Matching Problem: $\match$}
  \probleminput{A pattern $\alpha$, with $|\alpha|=m$, a word $w$, with $|w|=n$.}
  \problemquestion{Is there a substitution $h$ with $h(\alpha) = w$?}
\end{problemdescription}

$\match$ appears frequently in various areas of theoretical computer science, such as combinatorics on words (e.g., unavoidable patterns~\cite{Loth02}, string solving and the theory of word equations \cite{Loth97}), stringology (e.g., generalized function matching~\cite{ami:gen}), language theory (e.g., pattern languages~\cite{DBLP:journals/jcss/Angluin80}, the theory of extended regular expressions with backreferences~\cite{cam:afo,fri:mas,Fre2013,FreydenbergerSchmid2019}), database theory (e.g., the theory of document spanners~\cite{FreydenbergerHolldack2018,Freydenberger2019,FaginEtAl2015,SchmidSchweikardt2021,SchmidICDT2022,SchmidSchweikardtPODS2022}), or algorithmic learning theory (e.g., the theory of descriptive patterns for finite sets of words~\cite{shi:pat,DBLP:journals/jcss/Angluin80,DBLP:journals/tcs/FernauMMS18}).\looseness=-1

$\match$ is NP-complete \cite{DBLP:journals/jcss/Angluin80}, in general. In fact, a detailed analysis~\cite{ReidenbachS14,shi:pol2,FerSch2015,DBLP:journals/mst/FernauSV16,DBLP:journals/toct/FernauMMS20,schmid13} of the matching problem has provided a better understanding of the parameterized complexity of this problem, highlighting, in particular, several subclasses of patterns for which the matching problem is polynomial, when various structural parameters of patterns are bounded by constants. Prominent examples in this direction are patterns with a bounded number of repeated variables, patterns with bounded scope coincidence degree~\cite{ReidenbachS14}, patterns with bounded locality~\cite{DayFMN17}, or patterns with a bounded treewidth \cite{ReidenbachS14}. See \cite{DBLP:journals/toct/FernauMMS20,DayFMN17,ReidenbachS14} for efficient algorithms solving $\match_{P}$ restricted to (or, in other words, parameterized by) to such classes $P$ of patterns. In general, each of the structural parameters defining such classes $P$ is a number $k$ characterizing in some way the structure of the patterns of the class $P$ and the matching algorithms for the respective class of patterns runs in $O(n^{ck})$ for some constant $c$. Moreover, these restricted matching problems are usually shown to be $W[1]$-hard w.r.t. the respective parameters.  \looseness=-1

In \cite{mfcs2021}, the study of efficient matching algorithms for patterns with variables was extended to an approximate setting. More precisely, the problem of deciding, for a pattern $\alpha$ from a class of patterns $P$ (defined by structural restrictions), a word $w$, and a non-negative integer $\Delta$, whether there exists a substitution $h$ such that the Hamming distance $\hdist{h(\alpha)}{w}$ between $h(\alpha)$ and $w$ is at most $\Delta$ was investigated. The corresponding minimization problem of computing $\hdist{\alpha}{w}=\min\{\hdist{h(\alpha)}{w}\mid h$ is a substitution of the variables of $\alpha\}$ was also considered. The main results of \cite{mfcs2021} were rectangular time algorithms and matching conditional lower bounds for the class of regular patterns $\regPat$ (which contain at most one occurrence of any variable). Moreover, polynomial time algorithms were obtained for unary patterns (also known as one-variable patterns, which consist in one or more occurrences of a single variable, potentially interleaved with terminal strings) or non-cross patterns (which consist in concatenations of unary patterns, whose variables are pairwise distinct). However, as soon as the patterns may contain multiple variables, whose occurrences are interleaved, the problems became NP-hard, even if only one of the variables occurs more than once. As such, unlike the case of exact matching, the approximate matching problem under Hamming distance is NP-hard even if some of the aforementioned parameters (number of repeated variables, scope coincidence degree, treewidth, but, interestingly, not locality) were upper bounded by small constants. \looseness=-1

\textbf{Our Contribution.} In this paper, inspired by, e.g.,  \cite{landauVishkin,fasterPatternMatchingED,elasticDegenerateApprox,sanatizationEDImprv,sanatizationED,fastApproxmatchingUnified,dynamicStringAlignment,approxStringMatchingTour} where various stringology patterns are considered in an approximate setting under {\em edit distance} \cite{EditDistance1,Levenshtein65}, and as a natural extension of the results of \cite{mfcs2021}, we consider the aforementioned approximate matching problems (parameterized by a class of patterns $P$) for the edit distance $\edist{\cdot}{\cdot}$, instead of Hamming Distance:\looseness=-1
\begin{problemdescription}
  \problemtitle{Approximate Matching Decision Problem for $\misMatch_P$}
  \probleminput{A pattern $\alpha \in P$, with $|\alpha|=m$, a word $w$, with $|w|=n$, an integer $\Delta\leq m$.}
  \problemquestion{Is $\edist{\alpha}{w} \leq \Delta$?}
\end{problemdescription}

\begin{problemdescription}
  \problemtitle{Approximate Matching Minimisation Problem for $\minMisMatch_P$}
  \probleminput{A pattern $\alpha \in P$, with $|\alpha|=m$, a word $w$, with $|w|=n$.}
  \problemquestion{Compute $\edist{\alpha}{w}$.}
\end{problemdescription}

Our paper presents two main results, which allow us to paint a rather comprehensive picture of the approximate matching problem under edit distance. 

Firstly, we consider the class of regular patterns, and show that $\misMatch_\regPat$ and $\minMisMatch_\regPat$ can be solved in $O(n\Delta)$ time (where, for $\minMisMatch$, $\Delta$ is the computed result); a matching conditional lower bound follows from the literature \cite{editDistanceQuadraticHardness}. This is particularly interesting because the problem of computing $\edist{\alpha}{w}$ for $\alpha = w_0 x_1 w_1\ldots x_k w_k$ can be seen as the problem of computing the minimal edit distance between any string in which $w_1,\ldots,w_k$ occur, without overlaps, in this exact order and the word $w$. 

Secondly, we show that, unlike the case of matching under Hamming distance, $\misMatch_P$ becomes $W[1]$-hard already for $P$ being the class of unary patterns, with respect to the number of occurrences of the single variable. So, interestingly, the problem of matching patterns with variables under edit distance is computationally hard for all the classes (that we are aware of) of structurally restricted patterns with polynomial exact matching problem, as soon as at least one variable is allowed to occur an unbounded number of times. 

To complement the results presented in this paper, we note that, for the classes of patterns considered in \cite{DBLP:journals/toct/FernauMMS20,DayFMN17,ReidenbachS14,mfcs2021}, which admit polynomial-time exact matching algorithms, one can straightforwardly adapt those algorithms to work in polynomial time in the case of matching under edit distance, when a constant upper bound $k_1$ on the number of occurrences of each variable exists. The complexity of these algorithms is usually  $O(n^{f(k_1,k_2)})$, for a polynomial function $f$ and for $k_2$ being a constant upper bound for the value of the structural parameter considered when defining these classes (locality, scope coincidence degree, treewidth, etc.). If no restriction is imposed on the structure of the pattern, $\match$ (and, as such, the matching under both Hamming and edit distances) is NP-hard even if there are at most two occurrences of each variable \cite{FerSch2015}. 
\section{Preliminaries} \label{sec:Prel}

Some basic notations and definitions regarding strings and patterns with variables were already given in the introduction and,  for more details, we also refer to \cite{DBLP:conf/cwords/ManeaS19,mfcs2021}. 
We only recall here some further notations. The set of all patterns, over all terminal-alphabets $\Sigma$, is denoted $PAT = \bigcup_\Sigma PAT_\Sigma$. Given a word or pattern $\gamma$, we denote by $\alp(\gamma) = B$ the smallest set (w.r.t. inclusion) $B\subseteq \Sigma$ and by $\var(\gamma)=Y$ the smallest set $Y\subseteq \mathcal{X}$ such that $\gamma \in (B\cup Y)^\star$. For any symbol $t \in \Sigma \cup \mathcal{X}$ and $\alpha \in PAT_\Sigma$, $\len{\alpha}_t$ denotes the total number of occurrences of the symbol $t$ in $\alpha$. 
For a pattern $\alpha = w_0x_1w_1\ldots w_kx_k$, we denote by $\term(\alpha)=w_0w_1\ldots w_k$ the projection of $\alpha$ on the terminal alphabet $\Sigma$. \looseness=-1

For words $u,w \in \sig{\star}$, the \emph{edit distance} \cite{EditDistance1,Levenshtein65} between $u$ and $w$ is defined as the minimal number $\edist{u}{w}$ of letter insertions, letter deletions, and letter to letter substitutions which one has to apply to $u$ to obtain $w$ . 

We recall some basic facts about the edit distance. Assume that $u$ is transformed into $w$ by a sequence of edits $\gamma$ (i.e., $u$ is aligned to $w$ by $\gamma$). We can assume without losing generality that the edits in $\gamma$ are ordered left to right with respect to the position of $u$ where they are applied. Then, for each factorization $u=u_1\ldots u_k$ of $u$, there exists a factorization $w=w_1\ldots w_k$ of $w$ such that $w_i$ is obtained from $u_i$ when applying the edits of $\gamma$ which correspond to the positions of $u_i$, for $i\in\{1,\ldots,k\}$. Note that this factorization of $w$ is not unique: we assume that the insertions applied at the beginning of $u$ correspond to positions of $u_1$, the insertions applied at the end of $u$ correspond to positions of $u_k$, but the insertions applied between $u_{i-1}$ and $u_i$ can be split arbitrarily in two parts: when considering them in the order in which they occur in $\gamma$ (so left to right w.r.t. the positions of $u$ where they are applied) we assume to first have a (possibly empty) set of insertions which correspond to positions of $u_{i-1}$ and then a (possibly empty) set of insertions which correspond to positions of $u_{i}$. On the other hand, if $w=w_1w_2$, we can uniquely identify the shortest prefix $u_1$ (respectively, the longest prefix $u'_1$) of $u$ from which, when applying the edits of $\gamma$ we obtain the prefix $w_1$ of $w$. 

Now, for a pattern $\alpha$ and a word $w$, we can define the edit distance between $\alpha$ and $w$ as $\edist{\alpha}{w}= \min \{\edist{h(\alpha)}{w}\mid h\mbox{ is a substitution of the variables of }\alpha  \}$. It is worth noting that $\edist{\alpha}{w}\leq |w|+|\term(\alpha)|$. 

With these definitions, we can consider the two pattern matching problems for families of patterns $P\subseteq PAT$, as already defined in the introduction. In the first problem $\misMatch_P$, which extends $\match_P$, we allow for a certain edit distance $\Delta$ between the image $h(\alpha)$ of $\alpha$ under a substitution $h$ and the target word $w$ instead of searching for an exact matching. In the second problem, $\minMisMatch_P$, we are interested in finding the substitution $h$ for which the edit distance between $h(\alpha)$ and the word $w$ is minimal, over all possible choices of $h$.

As a remark, based on our general comments regarding the edit distance, the following theorem follows.
\begin{restatable}{theorem}{thmGeneral}\label{thm:general}
$\misMatch_{PAT}$ and $\minMisMatch_{PAT}$ can be solved in $O( n^{2 k_2 + k_1})$ time, where $k_1$ is the maximum number of occurrences of any variable in the input pattern $\alpha$ and $k_2$ is the total number of occurrences of variables in $\alpha$.
\end{restatable}


As mentioned in the Introduction, the result of the previous Theorem can be improved if we consider the two problems for classes of patterns with restricted structure, where we obtain algorithms whose complexity depends on the structural parameter associated to that class, rather than the total number of occurrences of variables.

\section{Our Results}

The first main result of our paper is about the class of regular patterns. A pattern $\alpha$ over the terminal alphabet $\Sigma$ is regular if $\alpha = w_0 (\Pi_{i=1}^{k} x_i w_i)$ where, for $i\in \{1,\ldots,k\}$, $w_i\in \Sigma^*$ and $x_i$ is a variable, and $x_i\neq x_j$ for all $i\neq j$. The class of regular patterns is denoted by $\regPat$. We can show the following theorem.

\begin{restatable}{theorem}{thmMisMatch}\label{thm:misMatch}
$\misMatch_{\regPat}$ can be solved in $O(n\Delta)$ time. For an accepted instance $w,\alpha, \Delta$ of $\misMatch_{\regPat}$ we also compute $\edist{\alpha}{w}$ (which is at most $\Delta$).
\end{restatable}
\begin{proof}
{\bf Preliminaries and setting.} We begin with an observation. For $\alpha = w_0 (\Pi_{i=1}^{k} x_i w_i)$, we can assume w.l.o.g. that $w_i\in \Sigma^+$ for all $i\leq k$ as otherwise we would have neighboring variables that could be replaced by a single variable; thus, $k\leq |\term(\alpha)|$. 
To avoid some corner cases, we can assume w.l.o.g. that $\alpha$ and $w$ start with the same terminal symbol (this can be achieved by adding a fresh letter $\$$ in front of both $\alpha$ and $w$). While not fundamental, these simplifications make the exposure of the following algorithm easier to follow.\looseness=-1

Before starting the presentation of the algorithm, we note that a solution for $\misMatch_{\regPat}$ with distance $\Delta=0$ is a solution to $\match_{\regPat}$ and can be solved in $\mathcal{O}(n)$ by a greedy approach (as shown, for instance, in \cite{DBLP:journals/toct/FernauMMS20}). Further, the special case $x_1w_1x_2$ can be solved by an algorithm due to Landau and Vishkin~\cite{landauVishkin} in $\mathcal{O}(n\Delta)$ time. In the following, we are going to achieve the same complexity for the general case of $\misMatch_{\regPat}$ by extending the ideas of this algorithm to accommodate the existence of an unbounded number of pairwise-distinct variables.

One important idea which we use in the context of computing the edit distance between an arbitrary regular pattern and a word is to interpret each regular variable as an arbitrary amount of ``free'' insertions on that position, where ``free'' means that they will not be counted as part of the actual distance (in other words, they do not increase this distance). Indeed, we can see that the factor which substitutes a variable should always be equal to the factor to which it is aligned (after all the edits are performed) from the target word, hence does not add anything to the overall distance (and, therefore, it is ``free''). As such, this factor can be seen as being obtained via an arbitrary amount of letter insertions. Now, using this observation, it is easier to design an $O(nm)$-time algorithm which computes the edit distance between the terminal words $\beta=\term(\alpha)$ (instead of the pattern $\alpha$) and $w$ with the additional property that, for the positions $F_g = \len{(\Pi_{i=0}^{g} |w_i|)}$ for $0 \leq g \leq k-1$, we have that the insertions done between positions $\beta[F_g]$ and $\beta[F_g +1]$ when editing $\beta$ to obtain $w$ do not count towards the total edit distance between $\beta$ and $w$. For simplicity, we denote the set $\set{F_g | 0 \leq g \leq k-1}$ by $F$, we set $F_k=+\infty$, and note that $|\beta|=m-k$ (so $\beta\in \Theta(m)$).\looseness=-1

The description of our algorithm is done in two phases. We first explain how $\misMatch_{\regPat}$ can be solved by dynamic programming in $O(nm)$ time. Then, we refine this approach to an algorithm which fulfills the statement of the theorem. 

When presenting our algorithms, we refer to an alignment of prefixes $\interval{\beta}{1}{j}$ of $\beta$ and $\interval{w}{1}{\ell}$ of $w$, which simply means editing $\interval{\beta}{1}{j}$ to obtain $\interval{w}{1}{\ell}$.\looseness=-1

\smallskip 

{\bf First phase: a classical dynamic programming solution.} We define the $(|\beta|+1)\times (n+1)$ matrix $D[\cdot][\cdot]$, where $D[j][\ell]$ is the edit distance between the prefixes $\interval{\beta}{1}{j}$, with $0 \leq j \leq |\beta|$, and $\interval{w}{1}{\ell}$, with $0 \leq \ell \leq n$, with the additional important property that the insertions done between positions $\beta[F_g]$ and $\beta[F_g +1]$, for $F_g\leq j$, are not counted in this distance (they correspond to variables in the pattern $\alpha$). As soon as this matrix is computed, we can retrieve the edit distance between $\alpha$ and $w$ from the element $D[m-k][n]$. Clearly, now the instance $(\alpha,w,\Delta)$ of $\misMatch_\regPat$ is answered positively if and only if $D[m-k][n]\leq \Delta$. So, let us focus on an algorithm computing this matrix.

The elements of the matrix $D[\cdot][\cdot]$ can be computed by dynamic programming in $O(mn)$ time (see Appendix). 
Moreover, by tracing back the computation of $D[m-k][n]$, we obtain a path consisting in elements of the matrix, leading from $D[0][0]$ to $D[m-k][n]$, which encodes the edits needed to transform $\beta$ into $w$. An edge between $D[j-1][\ell]$ and $D[j][\ell]$ corresponds to the deletion of $\beta[j]$; and edge between $D[j-1][\ell-1]$ and $D[j][\ell]$ corresponds to a substitution of $\beta[j]$ by $w[\ell]$, or to the case where $\beta[j]$ and $w[\ell]$ are left unchanged, and will be aligned in the end. Moreover, an edge between $D[j][\ell-1]$ and $D[j][\ell]$ corresponds to an insertion of $w[\ell]$ after position $j$ in $\beta$; this can be a free insertion too (and part of the image of a variable of $\alpha$), but only when $j\in F$. 
This concludes the first phase of our proof. \looseness=-1

\smallskip 

{\bf Second phase: a succinct representation and more efficient computation of the dynamic programming table.} In the second phase of our proof, we will focus on how to solve $\misMatch_{\regPat}$ more efficiently. The idea is to avoid computing all the elements of the matrix $D[\cdot][\cdot]$, and compute, instead, only the relevant elements of this matrix, following the ideas of the algorithm by Landau and Vishkin \cite{landauVishkin}. The main difference between the setting of that algorithm (which can be directly used to compute the edit distance between two terminal words or between a word $w$ and a pattern $\alpha$ of the form $xuy, xu$, or $uy$, where $x$ and $y$ are variables and $u$ is a terminal word) and ours is that, in our case, the diagonals of the matrix $D[\cdot][\cdot]$ are not non-decreasing (when traversed in increasing order of the rows intersected by the respective diagonal), as we now also have free insertions which may occur at various positions in $\beta$ (not only at the beginning and end). This is a significant complication, which we will address next.

The main idea of the optimization done in this second phase is that we could actually compute and represent the matrix $D[\cdot][\cdot]$ more succinctly, by only computing and keeping track of at most $\Delta$ relevant elements on each diagonal of this matrix, where relevant means that we cannot explicitly rule out the existence of a path leading from $D[0][0]$ to $D[m-k][n]$ which goes through that element. 

For the clarity of exposure, we recall that the diagonal $d$ of the matrix $D[\cdot][\cdot]$ is defined as the array of elements $D[j][\ell]$ where $\ell- j = d$ (ordered in increasing order w.r.t. the first component $j$), where $-|\beta|+1 \leq d \leq n$. Very importantly, for a diagonal $d$, we have that if $D[j][j+d]\leq D[j+1][j+1+d]$ then $ D[j+1][j+1+d]-D[j][j+d]\leq 1$; however, it might also be the case that $D[j][j+d]> D[j+1][j+1+d]$, when $D[j+1][j+1+d]$ is obtained from $D[j+1][j+d]$ by a free insertion. 

\smallskip 

{\bf Analysis of the diagonals, definition of $M_d[\delta]$ and its usage.} Now, for each diagonal $d$, with $-|\beta|+1 \leq d \leq n$, and $\delta\leq \Delta$, we define $M_d[\delta]= \max\{j\mid  D[j][j+d]=\delta,$ and $ D[j'][j'+d]>\delta$ for all $j'> j\}$ (by convention, $M_d[\delta]=-\infty$, if $\{j\mid  D[j][j+d]=\delta,$ and $ D[j'][j'+d]>\delta$ for all $j'> j\}=\emptyset$). That is, $M_d[\delta]$ is the greatest row where we find the value $\delta $ on the diagonal $d$ and, moreover, all the elements appearing on greater rows on that diagonal are strictly greater than $\delta$ (or $M_d[\delta]=-\infty$ if such a row does not exist). 

Note that if a value $\delta$ appears on diagonal $d$ and there exists some $j'$ such that $D[j][j+d]\geq \delta$ for all $j\geq j'$, then, due to the only relations which may occur between two consecutive elements of $d$, we have that $M_d[\delta]\neq -\infty$. In particular, if a value $\delta$ appears on diagonal $d$ then $M_d[\delta]\neq -\infty$ if and only if $D[|\beta|][|\beta|+d]\geq \delta$. Consequently, if there exists $k>0$ such that $M_d[\delta-k]=|\beta|$ then $M_d[\delta]= -\infty$. 

In general, all values $M_d[\delta]$ which are equal to $-\infty$ are not relevant to our computation. To understand which other values $M_d[\delta]$ are not relevant for our algorithm, we note that if there exist some $k> 0$ and $s\geq 0$ such that $M_{d+s}[\delta-k]= |\beta|$ then it is not needed to compute $M_{d-g}[\delta+h]$, for any $g,h\geq 0$, at all, as any path going from $D[0][0]$ to $D[|\beta|][n]$, which corresponds to an optimal sequence of edits, does not go through $D[M_{d-g}[\delta+h]][M_{d}[\delta+h]+d]$. If $s=0$, then it is already clear that $M_{d}[\delta]= -\infty$, and we do not need to compute it. If $s\geq 1$, it is enough to show our claim for $h=0$ and $g=0$. Indeed, assume that the optimal sequence of edits transforming $\beta$ into $w$ corresponds to a path from $D[0][0]$ to $D[|\beta|][n]$ going through $D[M_{d}[\delta]][M_{d}[\delta]+d]$. By the fact that $M_{d}[\delta]$ is the largest $j$ for which $D[j][j+d]\leq \delta$, we get that this path would have to intersect, after going through $D[M_{d}[\delta]][M_{d}[\delta]+d]$, the path from $D[0][0]$ to $D[M_{d+s}[\delta-k]][M_{d+s}[\delta-k]+d+s]=D[|\beta|][|\beta|+d+s]$ (which goes only through elements $\leq \delta-k$). As $k>0$, this is a contradiction, as the path from $D[0][0]$ to $D[|\beta|][n]$ going through $D[M_{d}[\delta]][M_{d}[\delta]+d]$ goes only through elements $\geq \delta$ after going through $D[M_{d}[\delta]][M_{d}[\delta]+d]$. So, $M_d[\delta]$ is not relevant if there exist $k>0$ and $s>0$ such that $M_{d+s}[\delta-k]=m$. 

Once all relevant values $M_d[\delta]$ are computed, for $d$ diagonal and $\delta \leq \Delta$, we simply have to check if $M_{n-|\beta|}[\delta] = |\beta|$ (i.e., $D[|\beta|][n]=\delta$) for some $\delta \leq \Delta$. So, we can focus, from now on, on how to compute the relevant elements $M_d[\delta]$ efficiently. In particular, all these elements are not equal to $-\infty$. 

\smallskip 

{\bf Towards an algorithm: understanding the relations between elements on consecutive diagonals.} Let us now understand under which conditions $D[j][\ell]=\delta$ holds, as this is useful to compute $M_d[\delta]$. In general, this means that there exists a path leading from $D[0][0]$ to $D[j][\ell]$ consisting only in elements with value $\leq \delta$, and which ends with a series of edges belonging to the diagonal $d=\ell-j$, that correspond to substitutions or to letters being left unchanged. In particular, if all the edges connecting $D[j'][j'+d]$ and $D[j][\ell]$ on this path correspond to unchanged letters, then $\interval{\beta}{j'}{j}$ is a common prefix of $\interval{\beta}{j'}{|\beta|}$ and $\interval{w}{j'+d}{n}$. Looking more into details, there are several cases when $D[j][\ell]=\delta$. \looseness=-1

If $j\notin F$ and $\beta[j]\neq w[\ell]$, then $D[j-1][\ell-1]\geq \delta-1$ and $D[j-1][\ell]\geq \delta-1$ and $D[j][\ell-1]\geq \delta-1$ and at least one of the previous inequalities is an equality (i.e., one of the following must hold: $D[j][\ell-1]=\delta-1$ or $D[j-1][\ell-1]=\delta-1$ or $D[j-1][\ell]=\delta-1$). If $j\notin F$ and $\beta[j]= w[\ell]$, then $D[j-1][\ell-1]\geq \delta$ and $D[j-1][\ell]\geq \delta-1$ and $D[j][\ell-1]\geq \delta-1$ and at least one of the previous inequalities is an equality.

If $j\in F$ and $\beta[j]\neq w[\ell]$, then $D[j-1][\ell-1]\geq \delta-1$ and $D[j-1][\ell]\geq \delta-1$ and $D[j][\ell-1]\geq \delta$ and at least one of the previous inequalities is an equality. If $j\in F$ and $\beta[j]= w[\ell]$ then $D[j-1][\ell-1]\geq \delta$ and $D[j-1][\ell]\geq \delta-1$ and $D[j][\ell-1]\geq \delta$ and at least one of the previous inequalities is an equality.

Moving forward, assume now that $M_d[\delta]=j\neq -\infty$. This means that $D[j][\ell]=\delta,$ and $ D[j''][j''+d]>\delta$ for all $j''> j$. By the observations above, there exists $j'\leq j$ such that $D[j'][j'+d]=\delta$ and the longest common prefix of $\interval{\beta}{j'}{|\beta|}$ and $\interval{w}{j'+d}{n}$ has length $j-j'+1$, i.e., it equals $\interval{\beta}{j'}{j}$. The last part of this statement means that once we have aligned $\interval{\beta}{1}{j'}$ to $\interval{w}{1}{j'+d}$, we can extend this alignment to an alignment of $\interval{\beta}{1}{j}$ to $\interval{w}{1}{j+d}$ by simply leaving the symbols of $\interval{\beta}{j'+1}{j}$ unchanged. 

Let us see now what this means for the elements of diagonals $d$, $d+1$, and $d-1$. 

Firstly, we consider the diagonal $d$. Here we have that $j'\geq M_{d}[\delta-1]+1$. Note that if $\delta-1$ appears on diagonal $d$ then $M_{d}[\delta-1]\neq -\infty$. 

Secondly, we consider the diagonal $d+1$. Here, for all rows $\ell$ with $j'\leq \ell\leq j$, we have that $D[\ell-1][\ell+d]\geq \delta-1$ and $D[j''-1][j''+d]>\delta-1$, for all $j''$ with $|\beta|\geq j''>j$. Therefore, if $\delta-1$ appears on diagonal $d+1$, either $D[m][m+d+1]\leq d-1$ or $M_{d+1}[\delta-1]\neq -\infty$ and $M_{d}[\delta-1]+1\leq j$.

Finally, we consider the diagonal $d-1$. Here, for all rows $\ell$ with $j'\leq \ell\leq j$, we have that $D[\ell][\ell+d-1]\geq \delta-1$ and $D[j''][j''+d-1]\geq \delta$, for all $j''$ with $m\geq j''>j$. Thus, either all elements on the diagonal $d-1$ are $\geq \delta$, or $\delta-1$ occurs on diagonal $d-1$ and $M_{d-1}[\delta-1]\neq -\infty$. In the second case, when $M_{d-1}[\delta-1]\neq -\infty$, we have that $j\geq M_{d-1}[\delta-1]$ as, otherwise, we would have that $D[M_{d-1}[\delta-1]][M_{d-1}[\delta-1]+d]\leq \delta$ and $M_{d-1}[\delta-1]>j$, a contradiction. 

Still on diagonal $d-1$, if $\delta$ occurs on it, then $M_d[\delta]\neq -\infty$ holds. So, for $g\leq k-1$ with $F_g\leq M_{d-1}[\delta]<F_{g+1}$, we have that $F_g\leq M_d[\delta]$. Indeed, otherwise we would have two possibilities. If the path connecting $D[0][0]$ to $D[M_{d-1}[\delta]][M_{d-1}[\delta]+d-1]$ via elements $\leq d$ intersects row $F_g$ on $D[F_g][F_g+d']$ for some $d'\leq d$, then $D[F_g][F_g +d]\leq D[F_g][F_g+d']\leq \delta$ and $F_g >j$, a contradiction. If the path connecting $D[0][0]$ to $D[M_{d-1}[\delta]][M_{d-1}[\delta]+d-1]$ via elements $\leq d$ intersects row $F_g$ on $D[F_g][F_g+d']$ for some $d'> d$, then the respective path will also intersect diagonal $d$ on a row $>j$ before reaching $M_{d-1}[\delta]$, a contradiction with the fact that $j$ is the last row on diagonal $d$ where we have an element $\leq \delta$. 


So, for $M_d[\delta]$ to be relevant, we must have $D[|\beta|][|\beta|+d+1]\geq \delta$ (so there exists no $k> 0$ such that $M_{d+1}[\delta-k]= |\beta|$). In this case, if $M_d[\delta]=j$, then the following holds. The path (via elements $\leq d$) from $D[0][0]$ to $D[j][j+d]$ goes through an element $D[g][g+d']=\delta-1$. If the last such element on the respective path is on diagonal $d$, then it must be $M_{d}[\delta-1]$. If it is on diagonal $d-1$, then either  $g=M_{d-1}[\delta-1]$ (and then the path moves on diagonal $d$ via an edge corresponding to an insertion) or $g<M_{d-1}[\delta-1]$ (and then the path moves on diagonal $d$ via an edge corresponding to an insertion); in this second case, we could replace the considered path by a path connecting $D[0][0]$ to $D[M_{d-1}[\delta-1]][M_{d-1}[\delta-1]+d-1]$ (via elements $\leq \delta-1$), which then moves on diagonal $d$ via an edge corresponding to an insertion, and continues along that diagonal (with edges corresponding to letters left unchanged). If $D[g][g+d']$ is on diagonal $d+1$ (i.e., $d'=d+1$) then, just like in the previous case, we can simply consider the path connecting $D[0][0]$ to $D[M_{d+1}[\delta-1]][M_{d+1}[\delta-1]+d+1]$ (via elements $\leq \delta-1$), which then moves on diagonal $d$ via an edge corresponding to a deletion, and then continues along diagonal $d$ (with edges corresponding to letters left unchanged). If $D[g][g+d']$ is on none of the diagonals $d-1,d,d+1$ then we reach diagonal $d$ by edges corresponding to free insertions from some diagonal $d''<d$. The respective path also intersects diagonal $d-1$ (when coming from $d''$ to $d$ by free insertions), so diagonal $d-1$ contains $\delta$ and $M_{d-1}[\delta]\neq \infty$, and we might simply consider as path between $D[0][0]$ and $D[j][j+d]$ the path reaching diagonal $d-1$ on position $D[F_g][F_g+d-1]$ (via elements $\leq \delta$), where $F_g\leq M_{d-1}[\delta]<F_{g+1}$, which then moves on diagonal $d$ by an edge corresponding to a free insertion, and then continues along $d$ (with edges corresponding to letters left unchanged, as $F_g$ is greater or equal to the row where the initial path intersected diagonal $d$). This analysis covers all possible cases.

\smallskip 

{\bf Computing $M_d[\delta]$.} Therefore, if $M_d[\delta]$ is relevant (and, as such, $M_d[\delta]\neq -\infty$), then $M_d[\delta]$ can be computed as follows. Let $g$ be such that $F_g\leq M_{d-1}[\delta]<F_{g+1}$ (and $g=-1$ and $F_g=-\infty$ if $ M_{d-1}[\delta]=-\infty$). Let $H=\max\{M_{d-1}[\delta-1], F_g, M_{d}[\delta-1]+1, M_{d+1}[\delta-1]+1\}$ (as explained, in the case we are discussing, at least one of these values is not $-\infty$). Then we have that $j\geq H$ and the longest common prefix of $\interval{\beta}{H+1}{|\beta|}$ and $\interval{w}{H+d+1}{n}$ is exactly $\interval{\beta}{H+1}{j}$ (or we could increase $j$). So, to compute $j=M_d[\delta]$, we compute $H$ and then we compute the longest common prefix $\interval{\beta}{H+1}{j}$ of $\interval{\beta}{H+1}{|\beta|}$ and $\interval{w}{H+d+1}{n}$. 

In general, $M_d[\delta]$ is not relevant either because there exists some $s\geq 0$ and $\delta'<\delta$ such that $M_{d+s}[\delta']=|\beta|$ or because all elements of diagonal $d$ are strictly greater than $\delta$. In the second case, we note that all values $M_{d-1}[\delta-1]$, $F_g$, $M_{d}[\delta-1]$, and $M_{d+1}[\delta-1]$ must be $-\infty$ (as otherwise the diagonal $d$ would contain an element equal to $\delta$), so our computation of $M_d[\delta]$ returns $-\infty$ (which is correct). 

Now, based on these observations, we can see a way to compute the relevant values $M_d[\delta]$, for $-|\beta|\leq d\leq n$ and $\delta \leq \Delta$ (without computing the matrix $D$). 

We first construct the word $\beta$ and longest common prefix data structures for the word $\beta w$, allowing us to compute $\LCP(\interval{\beta}{h}{|\beta|},\interval{w}{h+d}{n})$, the length of the longest common prefix of $\interval{\beta}{h}{|\beta|}$ and $\interval{w}{h+d}{n}$ for all $h$ and $d$. 

Then, we will compute the values of $M_d[0]$ for all diagonals $d$. Basically, we need to identify, if it exists, a path from $D[0][0]$ to $D[M_d[0]][M_d[0]+d]$ which consists only of edges corresponding to letters left unchanged, or to free insertions. By an analysis similar to the one done above, we can easily show that $M_{0}[0]$ is $\LCP(\interval{\beta}{1}{|\beta|},\interval{w}{1}{n})$ (which is $\geq 1$, by our assumptions). Further, $M[d][0]=-\infty$ for $d<0$ and, for $d\geq 0$, $M_{d}[0] = F_g + \LCP(\interval{\beta}{F_g+1}{|\beta|},\interval{w}{F_g+1+d}{n})$, where $F_g\in F$ is such that $F_g\leq M_{d-1}[0]< F_{g+1}$ ($M_{d}[0]=-\infty$ if such an element $F_g$ does not exist). 

Further, for $\delta$ from $1$ to $\Delta$ we compute all the values $M_d[\delta]$, in order for $d$ from $-|\beta|+1$ to $n$. We first compute the largest diagonal $d'$ such that $M_{d'}[\delta-k]=|\beta|$, for some $k>0$. We will only compute $M_d[\delta]$, for $d$ from $d'+1$ to $n$. For each such diagonal $d$, we compute $g$ such that $F_g\leq M_{d-1}[\delta]<F_{g+1}$ and $H=\max\{M_{d-1}[\delta-1], F_g, M_{d}[\delta-1]+1, M_{d+1}[\delta-1]+1\}$. Then we set $M_d[\delta]$ to be $H+\LCP(\interval{\beta}{H+1}{|\beta|},\interval{w}{H+d+1}{n})-1$. 

\smallskip 

{\bf Conclusions.} This algorithm, which computes all relevant values $M_d[\delta], $can be implemented in $O((n+m)\Delta)$ time, as discussed in the Appendix (where also its pseudocode is given). As explained before, this allows us to solve $\misMatch_{\regPat}$ for the input $(\alpha,w,\Delta)$. Moreover, if the instance can be answered positively, the value $\delta$ for which $M_{n-|\beta|}[\delta] = |\beta|$ equals $\edist{\alpha}{w}$.
\qed \end{proof}

The following result now follows.
\begin{restatable}{theorem}{thmMinMisMatch}\label{thm:minMisMatch}
$\minMisMatch_{\regPat}$ can be solved in $O(n\Phi)$ time, where $\Phi=\edist{\alpha}{w}$.
\end{restatable}

The upper bounds reported in Theorems \ref{thm:misMatch} and \ref{thm:minMisMatch} are complemented by the following conditional lower bound, known from the literature \cite[Thm.  3]{editDistanceQuadraticHardness} (see Appendix). 
\begin{restatable}{theorem}{regLowerBound}\label{thm:regLowerBound}
$\misMatch_{\regPat}$ can not be solved in time $\mathcal{O}(|w|^h \Delta^g)$ (or $\mathcal{O}(|w|^h |\alpha|^g)$) where $h+g= 2-\epsilon$ with $\epsilon>0$, unless the Orthogonal Vectors Conjecture fails. \looseness=-1
\end{restatable}
It is worth noting that the lower bound from Theorem \ref{thm:regLowerBound} already holds for very restricted regular patterns, i.e., for $\alpha=x u y$, where $u$ is a string of terminals and $x$ and $y$ are variables. Interestingly, a similar lower bound (for such restricted patterns) does not hold in the case of the Hamming distance, covered in \cite{mfcs2021}.  

Our second main result addresses another class of restricted patterns. To this end, we consider the class of unary (or one-variable) patterns $\oneVarPat$, which is defined as follows: $\alpha \in \oneVarPat$ if there exists $x\in X$ such that $\var(\alpha)= \{x\}$. An example of unary pattern is $\alpha_1=\mathtt{ab} x \mathtt{ab} xx \mathtt{baab}$. 

We can show the following theorem.
\begin{restatable}{theorem}{thmUnaryW}\label{thm:unaryW1}
$\misMatch_\oneVarPat$ is $W[1]$-hard w.r.t. the number of occurrences of the single variable $x$ of the input pattern $\alpha$.
\end{restatable}
\begin{proof}[Sketch]
We begin by recalling the following problem:
\begin{problemdescription}
    \problemtitle{Median String: $\MS$}
    \probleminput{$k$ strings $w_1, \ldots, w_k \in \sigma^*$ and an integer $\Delta$.}
    \problemquestion{Does there exist a string $s$ such that $\sum_{i=1}^k \edist{w_i}{s}\leq \Delta$?\newline (The string $s$ for which $\sum_{i=1}^k \edist{w_i}{s}$ is minimum is called the median string of the strings $\{w_1,\ldots,w_k\}$.)}
\end{problemdescription}
Without loss of generality, we can assume that $\Delta\leq \sum_{i=1}^k |w_i|$ as, otherwise, the answer is clearly yes (for instance, for $s=\varepsilon$ we have that $\sum_{i=1}^k \edist{w_i}{\varepsilon}\leq \sum_{i=1}^k |w_i|$). Similarly, we can assume that $|s|\leq \Delta + \max\{|w_i|\mid i\in \{1,\ldots,k\} \}.$ In \cite{weightedMedianStringHardness5} it was shown that $\MS$ is NP-complete even for binary input strings and W[1]-hard with respect to the parameter $k$, the number of input strings. 

We will reduce now $\MS$ to $\misMatch_\oneVarPat$, such that an instance of $\MS$ with $k$ input strings is mapped to an instance of $\misMatch_\oneVarPat$ with exactly $k$ occurrences of the variable $x$ (the single variable occurring in the pattern). 

Thus, we consider an instance of $\MS$ which consists in the $k$ binary strings $w_1, \ldots, w_k \in \{0,1\}^*$ and the integer $\Delta\leq \sum_{i=1}^k |w_i|$. 
The instance of $\misMatch_\oneVarPat$ which we construct consists of a word $w$ and a pattern $\alpha$, such that $\alpha$ contains exactly $k$ occurrences of a variable $x$, and both strings are of polynomial size w.r.t. the size of the $\MS$-instance. Moreover, the bound on the $\edist{\alpha}{w}$ defined in this instance of $\misMatch_\oneVarPat$ equals $\Delta$. That is, if there exists a solution for the $\MS$-instance such that $\sum_{i=1}^k \edist{w_i}{s}\leq \Delta$, then, and only then, we should be able to find a solution of the $\misMatch_\oneVarPat$-instance with $\edist{\alpha}{w}\leq \Delta$.
The construction of the $\misMatch_\oneVarPat$ instance is realized in such a way that the word $w$ encodes the $k$ input strings for $\MS$, conveniently separated by some long strings over two fresh symbols $\$,\#$, while $\alpha$ can be obtained from $w$ by simply replacing each of the words $w_i$ by a single occurrence of the variable $x$. Intuitively, in this way, for $\edist{\alpha}{w}$ to be minimal, $x$ should be mapped to the median string of $\{w_1,\ldots,w_k\}$.
In this proof sketch, we just define the reduction. The proof of its correctness is given in the Appendix. \looseness=-1

For the strings $w_1, \ldots w_k \in \{0,1\}^*$, let $S=6 (\sum_{i=1}^k|w_i|)$; clearly, $S\geq 6 \Delta$. Let $w=w_1(\$^S \#^S)^Sw_2(\$^S \#^S)^S \ldots w_k(\$^S \#^S)^S$ and $\alpha = \left( x(\$^S \#^S)^S \right)^k$. 

The constructed instance of $\misMatch_\oneVarPat$ (i.e., $\alpha, w,\Delta$) is of polynomial size w.r.t. the size of the $\MS$-instance (i.e., $\{w_1,\ldots,w_k\}, \Delta$). Therefore, it (and our entire reduction) can be computed in polynomial time. Moreover, we can show that the instance $(w,\alpha,\Delta)$ of $\misMatch_\oneVarPat$ is answered positively if and only if the original instance of $\MS$ is answered positively. Finally, as the number of occurrences of the variable $x$ blocks in $\alpha$ is $k$, where $k$ is the number of input strings in the instance of $\MS$, and $\MS$ is $W[1]$-hard with respect to this parameter, it follows that $\misMatch_\oneVarPat$ is also $W[1]$-hard when the number of occurrences of the variable $x$ in $\alpha$ is considered as parameter. The statement follows.
\qed \end{proof}

A simple corollary of Theorem \ref{thm:general} is the following:
\begin{restatable}{theorem}{thmUnaryUpperBound}\label{thm:unaryUpperBound}
$\misMatch_{\oneVarPat}$ and $\minMisMatch_{\oneVarPat}$ can be solved in $O( n^{3 |\alpha|_x})$ time, where $x$ is the single variable occurring in $\alpha$.
\end{restatable}

Clearly, finding a polynomial time algorithm for $\misMatch_{\oneVarPat}$, for which the degree of the polynomial does not depend on $|\alpha|_x$, would be ideal. Such an algorithm would be, however, an FPT-algorithm for $\misMatch_{\oneVarPat}$, parameterized by $|\alpha|_x$, and, by Theorem \ref{thm:unaryW1} and common parameterized complexity assumptions, the existence of such an algorithm is unlikely. This makes the straightforward result reported in Theorem~\ref{thm:unaryUpperBound} relevant, to a certain extent.
\section{Conclusion}
Our results regarding the problem $\misMatch$ for various classes of patterns are summarized in Table \ref{results}, which highlights the differences to the case of exact matching and to the case of approximate matching under Hamming distance. 

\begin{table}[H]
\scriptsize
{
\begin{center}
\caption{\footnotesize Our new results are listed in column 4. The results overviewed in column 3 were all shown in \cite{mfcs2021}. We assume $|w|=n$ and $|\alpha|=m$.}
\label{results}

\smallskip

\begin{tabular}{|l|l|l|l|}
\hline
Class & $\match(w,\alpha)$ & $\misMatch(w,\alpha,\Delta)$ & $\misMatch(w,\alpha,\Delta)$ \\ 
 & & for $\hdist{\cdot}{\cdot}$ & for $\edist{\cdot}{\cdot}$ \\ 
\hline
$\regPat$ &$O(n)$ [folklore] & $O(n\Delta)$, matching  & $O(n\Delta)$, matching  \\
 & & cond. lower bound & cond. lower bound \\
\hline
$\oneVarPat$ & $O(n)$ [folklore] & $O(n)$  & $O(n^{3|\alpha|_x})$\\
$(\var(\alpha)=\{x\})$ & & & W[1]-hard w.r.t. $|\alpha|_x$  \\
\hline
$\nonCrossPat$ & $O(nm\log n)$ \cite{DBLP:journals/toct/FernauMMS20}& $O(n^3p)$  & NP-hard \\
\hline
$\oneRepPat$ & $O(n^2)$ \cite{DBLP:journals/toct/FernauMMS20} & $O(n^{k+2}m)$ & NP-hard for $k\geq 1$ \\
$k$=\# $x$-blocks & & W[1]-hard w.r.t. $k$  &  \\
\hline
$\kLocPat$ & $O(mkn^{2k+1})$ \cite{DayFMN17} & $O(n^{2k+2}m)$ & NP-hard for $k\geq 1$ \\
 & W[1]-hard w.r.t. $k$ & W[1]-hard w.r.t. $k$ & \\
\hline
$\kScdPat$ & $O(m^2n^{2k})$ \cite{DBLP:journals/toct/FernauMMS20} & NP-hard for $k\geq 2$ & NP-hard for $k\geq 1$ \\
 & W[1]-hard w.r.t. $k$ & & \\
\hline
$\kRepPat$ & $O(n^{2k})$ \cite{DBLP:journals/toct/FernauMMS20}  & NP-hard for $k\geq 1$ & NP-hard for $k\geq 1$\\
 & W[1]-hard w.r.t. $k$ & & \\
\hline
$k$-bounded & $O(n^{2k+4})$ \cite{ReidenbachS14} & NP-hard for $k\geq 3$ & NP-hard for $k \geq 1$ \\
treewidth  & W[1]-hard w.r.t. $k$ & & \\
\hline
\end{tabular}
\vspace{-1cm}
\end{center}}
\end{table}

Note that the results reported in the first row of the rightmost column of this table are based on Theorem \ref{thm:misMatch} (the upper bound) and Theorem \ref{thm:regLowerBound} (the lower bound). The rest of the cells of that rightmost column are all consequences of the result of Theorem \ref{thm:unaryW1}. Indeed, the classes of patterns covered in this table, which are presented in detail in \cite{mfcs2021}, are defined based on a common idea. In the pattern $\alpha$, we identify for each variable $x$ the $x$-blocks: maximal factors of $\alpha$ (w.r.t. length) which contain only the variable $x$ and terminals, and start and end with $x$. Then, classes of patterns are defined based on the way the blocks defined for all variables occurring in $\alpha$ are interleaved. However, in the patterns of all these classes, there may exist at least one variable which occurs an unbounded number of times, i.e., they all include the class of unary patterns. Therefore, the hardness result proved for unary patterns carries over and, as the structural parameters used to define those classes do not take into account the overall number of occurrences of a variable, but rather the number of blocks for the variables (or the way they are interleaved), we obtain NP-hardness for $\misMatch$ for that class, even if the structural parameters are trivial. 

While our results, together with those reported in \cite{mfcs2021}, seem to completely characterize the complexity of $\misMatch$ and $\minMisMatch$ under both Hamming and edit distances, there are still some directions for future work. Firstly, in \cite{fineGrainedMedianString} the fine-grained complexity of computing the median string under edit distance for $k$ input strings is discussed. Their main result, a lower bound, was only shown for inputs over unbounded alphabets; it would be interesting to see if it still holds for alphabets of constant size. Moreover, it would be interesting to obtain similar lower bounds for $\misMatch_{\oneVarPat}$, as the two problem seem strongly related. To that end, it would be interesting if the upper bound of Theorem \ref{thm:unaryUpperBound} can be improved, and brought closer to the one reported for median string in \cite{sankoff75}. Secondly, another interesting problem is related to Theorem \ref{thm:regLowerBound}. The lower bound we reported in that theorem holds for regular patterns with a constant number of variables (e.g., two variables). It is still open what is the complexity of $\misMatch$ for regular patterns with a constant number of variables under Hamming distance. \looseness=-1

%
%
%
\bibliographystyle{splncs04}
\bibliography{references}

\newpage

\appendix

\section{Computational Model}
\label{compModel}
The computational model we use to describe our results is the standard unit-cost RAM with logarithmic word size: for an input of size $n$, each memory word can hold $\log n$ bits. Arithmetic and bitwise operations with numbers in $[1:n]$ are, thus, assumed to take $O(1)$ time. Numbers larger than $n$, with $\ell$ bits, are represented in $O(\ell/\log n)$ memory words, and working with them takes time proportional to the number of memory words on which they are represented. In all the problems, we assume that we are given a word $w$ and a pattern $\alpha$, with $|w|=n$ and $|\alpha|=m$, over a terminal-alphabet $\Sigma=\{1,2,\ldots,\sigma\}$, with $|\Sigma|=\sigma\leq n+m$. The variables are chosen from the set $\{x_1,\ldots, x_m\}$ and can be encoded as integers between $n+1$ and $n+m$. That is, we assume that the processed words are sequences of integers (called letters or symbols), each fitting in $O(1)$ memory words. This is a common assumption in string algorithms: the input alphabet is said to be {\em an integer alphabet}. For instance, the same assumption was also used for developing efficient algorithms for $\match$ in \cite{DBLP:journals/tcs/FernauMMS18,mfcs2021}. For a more detailed general discussion on this computational model see, e.g.,~\cite{crochemore}.

\section{Longest Common Prefix data structure (LCP)}

Given a word $w$, of length $n$, we can construct in $O(n)$-time {\em longest common prefix}-data structures which allow us to return in $O(1)$-time the value $LCP_w(i,j)=max\{|v|\mid v\mbox{ is a prefix of both } w[i:n]\mbox{ and }w[j:n]\}$. See \cite{DBLP:conf/icalp/KarkkainenS03,DBLP:journals/jacm/KarkkainenSB06} and the references therein. 
Now, given a word $w$, of length $n$, and a word $\beta$, of length $m$, we can construct in $O(n+m)$-time data structures which allow us to return in $O(1)$-time the value 
$\LCP(\interval{w}{i}{n},\interval{\beta}{j}{m})$, the length of the longest common prefix of the strings $\interval{\beta}{j}{m}$ and $\interval{w}{i}{n}$ for all $j$ and $i$.
In other words, $\LCP(\interval{w}{i}{n},\interval{\beta}{j}{m}) = max\{|v|\mid v\mbox{ is a prefix of both } w[i:n]\mbox{ and }\beta [j:m]\}$. This is achieved by constructing $LCP_{w\beta}$-data structures for the word $w\beta $, as above, and noting that $\LCP(\interval{w}{i}{n},\interval{\beta}{j}{m})=\min(LCP_{w\beta} (i,n+j),n-i)$. 
 
\section{Proofs} 

\subsection*{Proof of Theorem \ref{thm:general}}

\thmGeneral*
\begin{proof}
We only give the proof for $\minMisMatch_{PAT}$. 

Assume the input pattern is $\alpha = u_0 x_1 u_1 \ldots x_{k_2} u_{k_2}$ from $PAT_\Sigma$, where $x_i$ is a variable, for $i\in \{1,\ldots,{k_2}\}$, and $w_i\in \Sigma^*$ a terminal word, for $i\in \{0,\ldots,{k_2}\}$. Note that there might be the case that $x_i=x_j$ for some $i\neq j$, as there are no restrictions on the structure of the pattern $\alpha$. 

We make several observations. 

Let $h$ be a substitution of the variables from $\alpha$, such that $h(x_i)=t_i$, for $i\in \{1,\ldots,{k_2}\}$. Then, $h(\alpha)=u_0 t_1 u_1 \ldots t_{k_2} u_{k_2}$. When computing the edit distance $\edist{h(\alpha)}{w}$, one obtains a factorization of $w=w_0 w'_1 w_1\ldots w'_{k_2} w_{k_2}$ such that the optimal sequence of edits transforming $h(\alpha)$ into $w$ transforms $u_i$ into $w_i$, for $i\in \{0,\ldots,{k_2}\}$, and $t'_i$ into $w_i$, for $i\in \{1,\ldots,{k_2}\}$. 

Now, let $V_x=\{i\in \{1,\ldots,{k_2}\}\mid x_i=x\}$ and assume $h$ is a substitution of the variables from $\alpha$ such that $\edist{h(\alpha)}{w}$ is minimal w.r.t. all possible substitutions of the variables of $\alpha$. Moreover, let $h(x_i)=t_i$ for $i\in \{1,\ldots,{k_2}\}$. As before, there exists a factorization of $w=w_0 w'_1 w_1\ldots w'_{k_2} w_{k_2}$ such that the optimal sequence of edits transforming $h(\alpha)$ into $w$ transforms $u_i$ into $w_i$, for $i\in \{0,\ldots,{k_2}\}$, and $t'_i$ into $w_i$, for $i\in \{1,\ldots,{k_2}\}$. In this case, from the fact that $h$ is optimal, it is immediate that $h(x)=s_x$ where $s_x$ is the median string of $\{w'_i\mid i\in V_x\}$. 

Based on these observations, we can use the following algorithm solving $\minMisMatch_{PAT}$. 

For each $x\in \var(\alpha)$, define $V_x=\{i\in \{1,\ldots,{k_2}\}\mid x_i=x\}$. For each factorization $f$ of $w=w_0 w'_1 w_1\ldots w'_{k_2} w_{k_2}$ and for each variable $x$: compute the median string $s_x$ of $\{w'_i\mid i\in V_x\}$; define the substitution $h_f$ which maps $x$ to $s_x$ for all $x$; compute the edit distance $\edist{h_f(\alpha)}{w}$. After considering each possible factorization $f$, return the substitution $h_f$ for which $\edist{h_f(\alpha)}{w}$ is minimal. 

In the above algorithm, to compute the median string of $\{w'_i\mid i\in V_x\}$, we use the algorithm of \cite{sankoff75}. This algorithm runs in $O(\ell_x^{|V_x|})$, where $\ell_x= \max\{|w'_i|\mid i\in V_x\}$. Therefore, the running time of our algorithm can be upper bounded by $O(n^{2{k_2}} n^{k_1} )$, so also by $O(n^{2{k_2}+k_1} )$. 
\qed \end{proof}

\subsection*{Algorithms from the proof of Theorem \ref{thm:misMatch}}

\subsubsection*{Computing matrix $D[\cdot][\cdot]$.}

The elements of the matrix $D[\cdot][\cdot]$ can be computed by dynamic programming. The base cases are $D[j][0] = j$, for all $j\leq \beta$ and $D[0][\ell] = \ell$. In the case of computing $D[0][\ell]$, we simply insert all the letters of $\interval{w}{1}{\ell}$ in $\interval{\beta}{1}{0}=\varepsilon$, while in the case of $D[j][0]$ we are deleting all letters from $\interval{\beta}{1}{j}$ (and, if we refer to the edits in $\alpha$, where we also have variables, then we substitute all the variables of the prefix of $\alpha$ which corresponds to $\interval{\beta}{1}{j}$ by the empty word, as well). 

The rest of the elements of $D[\cdot][\cdot]$ are now computed according to two cases. 

Firstly, we consider the computation of $D[j][\ell]$ for $j \notin F$. In this case, we cannot use the aforementioned free insertions, so the element $D[j][\ell]$ is computed as in the case of computing the usual edit distance between two strings. 
        \begin{equation*}
              D[j][\ell] = min\begin{cases}
                D[j-1][\ell] + 1, & \text{$\beta[j]$ is deleted in the alignment of $\interval{\beta}{1}{j}$} \\
                $ $ & \text{to $\interval{w}{1}{\ell}$};\\
                D[j][\ell-1] + 1, & \text{$w[\ell]$ is inserted after position $j$ of $\beta$ in } \\
                $ $ & \text{the alignment of $\interval{\beta}{1}{j}$ to $\interval{w}{1}{\ell}$};\\
                D[j-1][\ell-1] + 1, & \text{$\beta[j]$ is substituted by $w[j]$ in} \\
                & \text{the alignment of $\interval{\beta}{1}{j}$ to $\interval{w}{1}{\ell}$};\\
                D[j-1][\ell-1], & \text{$\beta[j]$ is left unchanged in the alignment}\\
                $ $& \text{of $\interval{\beta}{1}{j}$ to $\interval{w}{1}{\ell}$ because $\beta[j]=w[\ell]$}.\\
              \end{cases}
        \end{equation*}

The more interesting case is when $j \in F$ and we can use free insertions. Naturally, our starting point is still represented by the four possible cases based on which we computed $D[j][\ell]$ when $j\notin F$. However, the case corresponding to the insertion of $w[\ell]$ to extend an alignment of $\interval{\beta}{1}{j}$ and $\interval{w}{1}{\ell-1}$ to an alignment of $\interval{\beta}{1}{j}$ and $\interval{w}{1}{\ell}$ can now be obtained by a free insertion, instead of an insertion of cost $1$. This brings us to the main difference between the two cases. In this case, an alignment between $\interval{\beta}{1}{j}$ and $\interval{w}{1}{\ell}$ can be obtained as follows. We first obtain an alignment of $\interval{\beta}{1}{j}$ to some prefix $\interval{w}{1}{\ell-k}$ of $w$ and then use free insertions to append $\interval{w}{\ell-k+1}{\ell}$ to the edited pattern, and, as such, obtain $\interval{w}{1}{\ell}$. But, this also means that we first obtain an alignment of $\interval{\beta}{1}{j}$ to some prefix $\interval{w}{1}{\ell-k}$ of $w$ and then use free insertions to append $\interval{w}{\ell-k+1}{\ell-1}$ to the edited pattern, and, as such, obtain an alignment of the pattern to $\interval{w}{1}{\ell-1}$, and then insert (again, without counting this towards the edit distance) $w[\ell]$ to obtain $\interval{w}{\ell-k+1}{\ell}$. Thus, in this case, an alignment between  $\interval{\beta}{1}{j}$ and $\interval{w}{1}{\ell}$ which uses free insertions corresponding to the position $j\in F$ is obtained from an alignment between  $\interval{\beta}{1}{j}$ and $\interval{w}{1}{\ell-1}$ followed by an additional free insertion. We obtain, as such, the following recurrence relation for $D[j][\ell]$, when $j\in F$: \looseness=-1
        \begin{equation*}
              D[j][\ell] = min\begin{cases}
                D[j-1][\ell] + 1, & \text{$\beta[j]$ is deleted};\\
                D[j-1][\ell-1] + 1, & \text{$\beta[j]$ is substituted by $w[\ell]$, if $\beta[j]\neq w[\ell]$};\\
                D[j-1][\ell-1], & \text{$\beta[j]$ is left unchanged, if $\beta[j]=w[\ell]$};\\
                D[j][\ell-1], & \text{$w[\ell]$ is inserted after position $j$, for free}.\\
              \end{cases}
        \end{equation*}

Using the two recurrence relation above, we can compute the elements of the matrix $D$ by dynamic programming (for $j$ from $0$ to $m-k$, for $\ell $ from $0 $ to $n$) in $O(nm)$ time. 

Moreover, by tracing back the computation of $D[m-k][n]$, we obtain a path consisting of elements of the matrix, leading from $D[0][0]$ to $D[m-k][n]$, which encodes the edits needed to transform $\beta$ into $w$. An edge between $D[j-1][\ell]$ and $D[j][\ell]$ corresponds to the deletion of $\beta[j]$; and edge between $D[j-1][\ell-1]$ and $D[j][\ell]$ corresponds to a substitution of $\beta[j]$ by $w[\ell]$, or to the case where $\beta[j]$ and $w[\ell]$ are left unchanged, and will be aligned in the end. Moreover, an edge between $D[j][\ell-1]$ and $D[j][\ell]$ corresponds to an insertion of $w[\ell]$ after position $j$ in $\beta$; this can be a free insertion too (and part of the image of a variable of $\alpha$), but only when $j\in F$. 

A listing of an algorithm computing $D[\cdot][\cdot]$ is given in Figure \ref{fig:firstalgo}. 

\begin{figure}[H]
    \centering
    \scalebox{.75}{
        \begin{algorithm}[H]
            \SetAlgoLined
            \KwData{$w, \alpha$}
            \KwResult{minimal edit distance between $\alpha$ and $w$ in $\mathcal{O}(nm)$}
            compute $\beta$ and the set $F$\;
            \For{$j\gets[0$ \KwTo $|\beta|]$}{
                $D[j][0]\gets j$\;
            }
            \For{$\ell\gets[0$ \KwTo $n]$}{
                $D[0][\ell]\gets \ell$\;
            }
            $g \gets 0\;$\\
            \For{$j\gets[0$ \KwTo $|\beta|]$}{
                \eIf{$j = F_g$}
                {
                    \For{$\ell\gets[0$ \KwTo $n]$}{
                    \[ 
                        D[j][\ell]\gets min 
                        \begin{cases} 
                            D[j][\ell-1],   & \text{free insertion}\\
                            D[j-1][\ell]+1,  & \text{deletion}\\ 
                            D[j-1][\ell-1]+1,& \text{substitution}\\
                            D[j-1][\ell-1],  & \text{if } w[\ell]=\beta[j]\\
                        \end{cases}\;
                    \]
                    }
                   $ g \gets g + 1\;$
                }{
                    \For{$\ell\gets[0$ \KwTo $n]$}{
                    \[ 
                        D[j][\ell]\gets min 
                        \begin{cases} 
                            D[j][\ell-1]+1,  & \text{insertion}\\
                            D[j-1][\ell]+1,  & \text{deletion}\\ 
                            D[j-1][\ell-1]+1,& \text{substitution}\\
                            D[j-1][\ell-1],  & \text{if } w[\ell]=\beta[j]\\
                        \end{cases}\;
                    \]
                        
                    }  
                }
            }
            return $D[|\beta|][n]$
        \end{algorithm}
    }
    \caption{Algorithm to compute $D[\cdot][\cdot]$ in $\mathcal{O}(nm)$ time.}
    \label{fig:firstalgo}
\end{figure}

\subsubsection*{Data structures for computing $M_d[\delta]$.}

We first use a linear time algorithm for the computation of the longest common prefix data structures for $\beta$ and $w$ (see the section of this Appendix about such data structures and \cite{DBLP:journals/jacm/KarkkainenSB06}). Secondly, we use an auxiliary array $G$ of size $|\beta|+1$, which stores for each positive integer $i\leq \beta$ the value $G[i]=\max \{g\mid F_g\leq i\}$, and can be computed in linear time. This allows us to efficiently retrieve the values $F_g$. Finally, while computing the values $M_d[\delta]$, for $d$ and $\delta$, we can maintain the value $d'$ of the greatest diagonal such that there exist $k$ with $M_{d'}[\delta-k]=|\beta|$: when we are done with computing all the values $M_d[\delta-1]$, for all $d$, we simply check if we need to update $d'$ because we might have found some $d''>d'$ for which $M_{d''}[\delta-1]=|\beta|$. 

\subsubsection*{Computing $M_d[\delta]$.}
The algorithm for computing the relevant values $M_d[\Delta]$ and how these are used to solve Problem $\misMatch_\regPat$ is given in Figure \ref{fig:diagalgo}.

\begin{figure}[H]
    \centering
    \scalebox{.75}{
    \begin{algorithm}[H]
            \SetAlgoLined
            \KwData{$w, \alpha$}
            \KwResult{minimal edit distance between $\alpha$ and $w$ in $\mathcal{O}((n+m)\Delta)$}
            construct $\beta$\;
            construct $F$\;
            init $g \gets -1$\;
            construct $LCP_{\beta, w}$\;
            \For{$d \gets [-\len{\beta}$ \KwTo $0]$}{
                $M_d[0] \gets -\infty$\;
            }
            $M_0[0] \gets LCP(\interval{\beta}{1}{\len{\beta}},\interval{w}{1}{n})$\;
            compute $g$ such that $F_g \leq M_0[0] < F_{g+1}$ ($g\gets -1$ if $F_0>M_0[d]$)\;
            \eIf{$g = -1$}{
                \For{$d \gets [1$ \KwTo $n]$}{
                    $M_d[0] \gets -\infty$\;
                }
            }
            {
                \For{$d \gets [1$ \KwTo $n]$}{
                    $M_d[0] \gets F_g + \LCP(\interval{\beta}{F_g+1}{|\beta|},\interval{w}{F_g+1+d}{n})$\;
                    update $g$ such that $F_g \leq M_d[0] < F_{g+1}$\;
                }
            }
            $g \gets -1$\; 
            compute $d'=\min\{d\leq n\mid M_d[0]=m\}$\; 
            \For{$\delta \gets[1$ \KwTo $\Delta]$}{
                \For{$d \gets[d'+1$ \KwTo $n]$}{
                    update $g$ such that $F_g \leq M_{d-1}[\delta] < F_{g+1}$\;
                    \[
                        H \gets max 
                            \begin{cases}
                                M_{d-1}[\delta-1],  & \text{diagonal below}\\ 
                                F_g, & \text{for $F_g$ with $F_g \leq M_{d-1}[\delta] < F_{g+1}$} \hspace{0.5cm}; \\
                                M[d][\delta-1] + 1, & \text{same diagonal}\\
                                M[d+1][\delta-1]+1,  & \text{diagonal above}\\
                            \end{cases}
                    \] 
                    $M[d][\delta] \gets H+\LCP(\interval{\beta}{H+1}{|\beta|},\interval{w}{H+d+1}{n})-1$\;
                    \If{$(d = n-|\beta|) ~\wedge~ (M[d][\delta] = |\beta|)$}{
                        \Return $\delta$\;
                    }
                }        
                maintain $d'=\min\{d''\leq |\beta| \mid M_{d''}[\delta-s]=|\beta|$ for some $s\geq 0\}$\; 

            }
            \Return No solution with $\Delta$ edit operations.\;
        \end{algorithm}
    }
    \caption{Algorithm to compute the relevant values of $M$ in $\mathcal{O}((n+m)\Delta)$ time.}
    \label{fig:diagalgo}
\end{figure}

\subsection*{Proof of Theorem \ref{thm:minMisMatch}}
\thmMinMisMatch*
\begin{proof}
We use the algorithm of Theorem \ref{thm:misMatch} for $\Delta=2^i$, for increasing values of $i$ starting with $1$ and repeating until the algorithm returns a positive answer and computes $\Phi=\edist{\alpha}{w}$. The algorithm is clearly correct. Moreover, the value of $i$ which was considered last is such that $2^{i-1}< \Phi\leq 2^{i}$. So $i=\lceil \log_2 \Phi\rceil$, and the total complexity of our algorithm is $O(n\sum_{i=1}^{\lceil \log_2 \Phi\rceil}2^i)=O(n\Phi)$.
\qed \end{proof}

\subsection*{Lower bound for $\misMatch_{\regPat}$}
The results of Theorems \ref{thm:misMatch} and \ref{thm:minMisMatch} are complemented by the following lower bound, known from the literature \cite{editDistanceQuadraticHardness}.
Firstly, we recall the $\OV$ problem.
\begin{problemdescription}
  \problemtitle{Orthogonal Vectors (for short, $\OV$)}
  \probleminput{ Two sets $U,V$ consisting each of $n$ vectors from $\set{0,1}^d$, where $d\in \omega(\log n)$.}
  \problemquestion{Do vectors $u \in U, v \in V$ exist, such that $u$ and $v$ are orthogonal, i.e., for all $1 \leq k \leq d$, $v[k] u[k] = 0$ holds?}
\end{problemdescription}

It is clear that, for input sets $U$ and $V$ as in the above definition, one can solve $\OV$ trivially
in $\mathcal{O}(n^2d)$ time. The following conditional lower bound is known.
\begin{restatable}[$\OV$-Conjecture]{lemma}{OVC}
$\OV$ can not be solved in $\mathcal{O}(n^{2-\epsilon} d^{c})$ for any $\epsilon > 0$ and constant $c$, unless the Strong Exponential Time Hypothesis (SETH) fails.  
\end{restatable}
See \cite{DBLP:conf/stacs/Bringmann19,DBLP:journals/tcs/Williams05} and the references therein for a detailed discussion regarding conditional lower bounds related to OV. In this context, the following result is an immediate consequence of \cite[Thm.  3]{editDistanceQuadraticHardness}.
\regLowerBound*

\subsection*{Proof of Theorem \ref{thm:unaryW1}}
\thmUnaryW*
Before starting the proof of Theorem \ref{thm:unaryW1} we need the following technical lemma.
\begin{restatable}{lemma}{lemReductionHelp}\label{lem:reductionHelp}
Let $\$$ and $\#$ be two letters and let $S$, $g$, and $\ell$ be integers. If $g\geq 0$, $2g\leq S$, $\frac{S}{2}\leq \ell-g$, and $\ell\leq S$ then:
\begin{enumerate}
    \item $\edist{\$^g (\$^S\#^S)^{S-1} \$^S\#^\ell}{(\$^S\#^S)^{S}}=g+(S-\ell);$
    \item $\edist{\$^{\ell}\#^S(\$^S\#^S)^{S-1}\#^g }{(\$^S\#^S)^S}=g+(S-\ell).$
\end{enumerate}

\end{restatable}
\begin{proof}
We only show the first claim, as the second follows identically (as it is symmetrical). 

Firstly, it is clear that $g+(S-\ell)$ edits suffice to transform $\$^g (\$^S\#^S)^{S-1} \$^S\#^\ell $ into $(\$^S\#^S)^{S}$.

Now, we will show that we cannot transform $\$^g (\$^S\#^S)^{S-1} \$^S\#^\ell $ into $(\$^S\#^S)^{S}$ with fewer than $S-\ell + g$ edits. 

Note that from $\frac{S}{2}\leq \ell-g$ we get $\ell \geq g+(S-\ell)$. So, the suffix $\#^S $ of $(\$^S\#^S)^{S}$ must be obtained by a series of edits from a suffix of the suffix $\$^S\#^\ell$ of $\$^g (\$^S\#^S)^{S-1} \$^S\#^\ell $. This means that at least $S-\ell$ edits must be performed in the respective suffix to obtain $S$ symbols $\#$. This leaves us with at most $g$ edits remaining to obtain $(\$^S\#^S)^{S-1} \$^S $. In particular, this means that the prefix $\$^S\#^S$ of $(\$^S\#^S)^{S-1} \$^S $ must be obtained from a prefix of the prefix $ \$^g \$^S \#^S$ of $\$^g (\$^S\#^S)^{S-1} \$^S\#^\ell $. As, in the best case, $g$ $\$$ symbols need to be substituted or removed, it follows that we need to use $g$ edits to obtain the prefix $\$^S\#^S$ of $(\$^S\#^S)^{S}$. As such, we already had to use $g+(S-\ell)$ edits to transform $\$^g (\$^S\#^S)^{S-1} \$^S\#^\ell $ into $(\$^S\#^S)^{S}$, so it cannot be done with fewer edits. The conclusion follows.
\qed \end{proof}

We can now proceed with the proof of Theorem \ref{thm:unaryW1}. 

\begin{proof}
{\bf Preliminaries.}
We begin by recalling the following problem:
\begin{problemdescription}
    \problemtitle{Median String: $\MS$}
    \probleminput{$k$ strings $w_1, \ldots, w_k \in \Sigma^*$ and an integer $\Delta$.}
    \problemquestion{Does there exist a string $s$ such that $\sum_{i=1}^k \edist{w_i}{s}\leq \Delta$?\newline (The string $s$ for which $\sum_{i=1}^k \edist{w_i}{s}$ is minimum is called the median string of the strings $\{w_1,\ldots,w_k\}$.)}
\end{problemdescription}
Without loss of generality, we can assume that $\Delta\leq \sum_{i=1}^k |w_i|$ as, otherwise, the answer is clearly yes (for instance, for $s=\varepsilon$ we have that $\sum_{i=1}^k \edist{w_i}{\varepsilon}\leq \sum_{i=1}^k |w_i|$). Similarly, we can assume that $|s|\leq \Delta + \max\{|w_i|\mid i\in \{1,\ldots,k\} \}.$

In \cite{weightedMedianStringHardness5} it was shown that $\MS$ is NP-complete even for binary input strings and W[1]-hard with respect to the parameter $k$, the number of input strings. 

\smallskip

{\bf Reduction: intuition and definition.}
We will reduce $\MS$ to $\misMatch_\oneVarPat$, such that an instance of $\MS$ with $k$ input strings is mapped to an instance of $\misMatch_\oneVarPat$ with exactly $k$ occurrences of the variable $x$ (the single variable occurring in the pattern). 

Thus, we consider an instance of $\MS$ which consists in the $k$ binary strings $w_1, \ldots, w_k \in \{0,1\}^*$ and the integer $\Delta$. As mentioned above, we can assume that in this instance $\Delta\leq \sum_{i=1}^k |w_i|$. 

The instance of $\misMatch_\oneVarPat$ which we construct consists of a word $w$ and a pattern $\alpha$, such that $\alpha$ contains exactly $k$ occurrences of a variable $x$, and both strings are of polynomial size w.r.t. the size of the $\MS$-instance. Moreover, the bound on the $\edist{\alpha}{w}$ defined in this instance equals $\Delta$. That is, if there exists a solution for the $\MS$-instance such that $\sum_{i=1}^k \edist{w_i}{s}\leq \Delta$, then, and only then, we should be able to find a solution of the $\misMatch_\oneVarPat$-instance with $\edist{\alpha}{w}\leq \Delta$. 

The construction of the $\misMatch_\oneVarPat$ instance is realized in such a way that the word $w$ encodes the $k$ input strings, conveniently separated by some long strings over $\{\$,\#\}$ (where $\$,\#$ are two fresh symbols), while $\alpha$ can be obtained from $w$ by simply replacing each of the words $w_i$ by a single occurrence of the variable $x$. Intuitively, in this way, for $\edist{\alpha}{w}$ to be minimal, $x$ should be mapped to the median string of $\{w_1,\ldots,w_k\}$.

We can now formally define the reduction.

For the $k$ binary strings $w_1, \ldots w_k \in \{0,1\}^*$defining the instance of $\MS$, let $S=6 (\sum_{i=1}^k|w_i|)$; clearly $S\geq 6 \Delta$. Let now $w=w_1(\$^S \#^S)^Sw_2(\$^S \#^S)^S \ldots w_k(\$^S \#^S)^S$ and $\alpha = \left( x(\$^S \#^S)^S \right)^k$. 

\smallskip

{\bf Reduction: correctness.} We prove first the correctness of the reduction, that is, the following claim: the instance of $\MS$ defined by $w_1,\ldots,w_k$ and $\Delta$ is answered positively if and only if the instance of $\misMatch_\oneVarPat$ defined by $w,\alpha,\Delta$ is answered positively.

Assume first that the instance of $\MS$ defined by $w_1,\ldots,w_k$ and $\Delta$ is answered positively. Then, it is immediate to see that $\edist{\alpha}{w}\leq \Delta$. Indeed, let $w'= \left( s(\$^S \#^S)^S \right)^k$ be the word obtained from $\alpha $ by replacing $x$ with the median string $s$ of $w_1,\ldots,w_k$. Then, clearly, $\edist {w'}{w}\leq \Delta$. 

Now, assume that the instance of $\misMatch_\oneVarPat$ defined by $w,\alpha,\Delta$ is answered positively. This means that there exists some word $t\in\{0,1,\$,\#\}^*$ such that $\edist{u}{w}\leq \Delta$ for $u=\left( t(\$^S \#^S)^S \right)^k$. 

Therefore, there exists an optimal (w.r.t. length) sequence of edits $\gamma$ which transforms $u$ into $w$, such that the length of $\gamma$ is at most $\Delta$. As explained in the preliminaries, we can assume that the edits in the sequence $\gamma$ are ordered increasingly by the position of $u$ to which they are applied (i.e., left to right). Our road plan is to show that if such a sequence of edits $\gamma$ exists, then there exists a sequence $\delta$ of edits of equal length (so also optimal) transforming $u$ into $w$, such that the edits rewrite the $i^{th}$ occurrence of the factor $t$ in $w$ into $w_i$, for $i$ from $1$ to $k$, and leave the rest of the string $u$ unchanged.

Let $u_1$ be the shortest prefix of $u$ from which we obtain the prefix $w_1(\$^S\#^S)^S$ of $w$ when applying the edits of $\gamma$. Clearly, $|w_1(\$^S\#^S)^S| - S\leq |u_1|\leq |w_1(\$^S\#^S)^S| + S $ (as the overall distance between $u$ and $w$ is upper bounded by $\Delta \leq S$). Let now $u'_1$ be the longest prefix of $u_1$ from which we obtain $w_1$ when applying the edits of $\gamma$, and let $u_1=u'_1u''_1$. Clearly, the edits of $\gamma$ transform $u''_1$ into $(\$^S\#^S)^S$. We are now performing a case analysis. 

\noindent {\bf Case 1:} $|u'_1| \leq |t|$.

\noindent {\bf Case 1.1:} $u''_1=v (\$^S\#^S)^S s $, where $v,s\in \{0,1,\#,\$\}^*$ and $v$ is a suffix of $t$ and $s$ a prefix of $(t(\$^S\#^S)^{S})^{k-1}$. As $|u''_1|=2S^2 + |v|+|s|$, then at least $|v|+|s|$ edits are needed to transform $u''_1$ into $(\$^S\#^S)^S$. We can modify $\gamma $ such that these operations are deletions of all symbols of $v$ and $s$, and obtain a new sequence of edits $\gamma'$. 

\noindent {\bf Case 1.2:} $u''_1=v (\$^S\#^S)^{S-1} s $, where $v\in \{0,1,\#,\$\}^*$ is a suffix of $t$ and $s=\$^S\#^\ell$ for some $\ell\geq S-\Delta$. We thus have $t=u'_1v$ and $|v|\leq 2\Delta $ (because $\left | |u''_1|-2S^2\right | \leq \Delta $). Further, when applying the operations of $\gamma$, after all the edits in $u_1$ were performed, we obtain $w_2(\$^S\#^S)^S\ldots w_k(\$^S\#^S)^S$ from $\#^{S-\ell}(t(\$^S\#^S)^{S})^{k-1}$ optimally. Hence, from $u''_1$ we obtain $(\$^S\#^S)^{S}$ so, after performing the $p$ edits corresponding to positions of $v$ (excluding the potential insertions on positions occurring to the right of the last symbol of $v$), we must edit them into $\$$ letters, so we must obtain a string $\$^g (\$^S\#^S)^{S-1} \$^S\#^\ell $ for some $0\leq g\leq 2 \Delta$. It is immediate that $p+g\geq |v|$ (as when counting the $p$ edit operations, we count the symbols which were deleted from $v$, while all the symbols which were substituted in $v$ correspond to distinct positions of $\$^g$). Now, by Lemma \ref{lem:reductionHelp}, since $g\leq 2 \Delta$, $S-\ell \leq \Delta$, and $S\geq 6\Delta$, we get that the minimum number of edits needed to transform $\#^g (\$^S\#^S)^{S-1} \#^S\$^\ell $ into $(\$^S\#^S)^{S} $ is $g+(S-\ell)$. So, to transform $u''_1$ into $(\$^S\#^S)^{S} $ we use $p+g+S-\ell\geq |v|+S-\ell$ edits. We can, therefore, modify $\gamma $ to obtain a new sequence of edits $\gamma'$, which has at most the same length as $\gamma$, in which we first apply all the edit operations from $\gamma$ to $u'_1$, then we delete all symbols of $v$, then we simply leave $(\$^S\#^S)^{S}$ alone, then we insert $\#^{S-\ell}$ after $(\$^S\#^S)^{S}$, and we continue by editing $\#^{S-\ell}(t(\$^S\#^S)^{S})^{k-1}$ into $w_2(\$^S\#^S)^S\ldots w_k(\$^S\#^S)^S$ exactly as in $\gamma$. Clearly, we have just replaced $p+g+S-\ell$ operations in $\gamma$ by $|v|+S-\ell$ edits to obtain $\gamma'$. As $\gamma$ was of optimal length, and $p+g+S-\ell \geq |v|+S-\ell$, we have that $\gamma'$ must be of optimal length too. 

\noindent {\bf Case 2:} $|u'_1| > |t|$. Then $u'_1= t\$^{S-\ell} $, for some $\ell$ such that $0<S-\ell \leq \Delta$. 

\noindent {\bf Case 2.1:} $u''_1= \$^{\ell }\#^S(\$^S\#^S)^{S-2} \$^S \#^{S-g} $ for some $g$ such that $(S-\ell) + g \leq \Delta$. Moreover, when considering the sequence $\gamma$, we have that $\#^g (t (\$^S\#^S)^{S})^{k-1}$ is transformed into $w_2(\$^S\#^S)^{S}\ldots w_k(\$^S\#^S)^{S}$ optimally after the edits in $u_1$ are performed. As $|u''_1|=2S^2 - g - (S - \ell)$, then at least $S - \ell + g$ edits are needed to transform $u''_1$ into $(\$^S\#^S)^S$. Now we can modify $\gamma$ as follows. We first note that, in $\gamma$, the suffix $\$^{S-\ell}$ of $u'_1$ has to be completely rewritten to obtain $w_1$ (as $w_1$ does not contain $\$ $ symbols). Therefore, we  transform $t$ into $w_1$ by simulating the edits performed in the suffix $\$^{S-\ell} $ by only applying insertions after the last symbol of $t$ (instead of substitutions in $\$^{S-\ell}$ we do insertions, the insertions are done as before, and the deletions from $\$^{S-\ell}$ are not needed anymore); the number of these insertions is at most as big as the number of initial edits applied to the suffix $\$^{S-\ell}$ of $u'_1$. Then, the factor $(\$^S \#^S)^S$ following the first $t$ in $u$ is not edited, as it corresponds to the identical factor of $w$ which follows $w_1$, and then we insert after the first factor $(\$^S \#^S)^S$ of $u$ a factor $\#^g$, with $g$ insertions, and then continue editing $\#^g (t (\$^S\#^S)^{S})^{k-1}$ to obtain $w_2(\$^S\#^S)^{S}\ldots w_k(\$^S\#^S)^{S}$ as in $\gamma$. The resulting sequence $\gamma'$ of edits is at least ${S-\ell}$ edits shorter than $\gamma$, with $S-\ell >0$. As $\gamma$ was optimal, this is a contradiction, so this case is not possible.  

\noindent {\bf Case 2.2:} $u''_1= \$^{\ell}\#^S (\$^S\#^S)^{S-1} s $ where $0<S- \ell \leq \Delta$ and $s\in \{0,1,\#,\$\}^*$ is a prefix of $(t(\$^S\#^S)^{S})^{k-1}$. In this case, in $\gamma$, we have that the suffix $u'$ of $u$ occurring after $u_1$ is transformed into $w_2(\$^S\#^S)^{S}\ldots w_k(\$^S\#^S)^{S}$ optimally after the edits in $u_1$ are performed. Now, the suffix $s$ is transformed, by $p$ edits into $\#^g$ for some $g \leq 2 \Delta $, and we have $p+g \geq |v|$ (similarly to the Case 1.2). By Lemma \ref{lem:reductionHelp}, as in Case 1.2, we get that the minimum number of edits needed to transform $\$^\ell\#^S(\$^S\#^S)^{S-1}\#^g $ into $(\$^S\#^S)^S$ is $g+(S-\ell)$. So, overall, the number of edits needed to transform $\$^{\ell}\#^S (\$^S\#^S)^{S-1} s $ into $(\$^S\#^S)^S$ is $(S-\ell)+g+p\geq (S-\ell)+|s|$. Therefore, we can modify $\gamma$ as follows to obtain a new optimal sequence of edits $\gamma'$. As in Case 2.1 we simulate the edits in the suffix $\$^{S-\ell} $ of $u'_1$ by insertions. Then, the factor $(\$^S\#^S)^{S}$ is left unchanged. Then we simply delete the letters of $s$, and we continue by editing $u'$ as in $\gamma$ to obtain $w_2(\$^S\#^S)^{S}\ldots w_k(\$^S\#^S)^{S}$. Clearly, in $\gamma'$ we have at least $S-\ell$ edits less than in $\gamma$, with $S-\ell >0$. As $\gamma$ was optimal, his is a contradiction, so this case is also not possible.  

This concludes our case analysis.

In all possible cases (1.1 and 1.2), in the newly obtained sequence $\gamma'$ of edits, which has the same optimal length as $\gamma$, we have that the prefix $t$ of $u$ is transformed into $w_1$ by a sequence of edits $\gamma'_1$ (which ends with the deletion of the suffix $v$ of $t$), the first factor $(\$^S\#^S)^S$ of $u$ is then trivially transformed  (by an empty sequence of edits) into the first factor $(\$^S\#^S)^S$ of $w$, and then $(t(\$^S\#^S)^{S})^{k-1}$ is transformed into $w_2(\$^S\#^S)^S\ldots w_k(\$^S\#^S)^S$ by an optimal sequence of edits $\gamma'_2$ (which starts, in Case 1.1, the deletion of the prefix $s$ of $(t(\$^S\#^S)^{S})^{k-1}$ or, in Case 1.2, with the insertion of a factor $\#^{S-\ell}$ before $(t(\$^S\#^S)^{S})^{k-1}$). 

Now, we can apply the same reasoning, inductively, to the optimal sequence of edits $\gamma'_2$ which transforms $(t(\$^S\#^S)^{S})^{k-1}$ into $w_2(\$^S\#^S)^S\ldots w_k(\$^S\#^S)^S$, and, we will ultimately obtain that there exists an optimal sequence of edits $\delta$ which transforms $u$ into $w$ by transforming the $i^{th}$ factor $t$ of $u$ into $w_i$, for all $i$ from $1$ to $k$, and leaving the rest of the symbols of $u$ unchanged. As the length of $\delta$ is at most $\Delta$, this means that for the string $t$ we have $\sum_{i=1}^k\edist{t}{w_i}\ leq \Delta$, so the instance defined by $w_1,\ldots,w_k$ and $\Delta$ of $\MS$ can be answered positively. This concludes the proof of our claim and, as such, the proof of the correctness of our reduction.

\smallskip

{\bf Conclusion.} 
The instance of $\misMatch_\oneVarPat$ (i.e., $\alpha, w,\Delta$) is of polynomial size w.r.t. the size of the $\MS$-instance. Therefore, the instance of $\minMisMatch_{\oneRepPat}$ can be computed in polynomial time, and our entire reduction is done in polynomial time. Moreover, we have shown that the instance $(w,\alpha,\Delta)$ of $\misMatch_\oneVarPat$ is answered positively if and only if the original instance of $\MS$ is answered positively. Finally, as the number of occurrences of the variable $x$ blocks in $\alpha$ is $k$, where $k$ is the number of input strings in the instance of $\MS$, and $\MS$ is $W[1]$-hard with respect to this parameter, it follows that $\misMatch_\oneVarPat$ is also $W[1]$-hard when the number of occurrences of the variable $x$ in $\alpha$ is considered as parameter. This completes the proof of our theorem.
\qed \end{proof}

\end{document}